\documentclass{amsart}
\usepackage{amsmath,graphicx}

\usepackage{amsthm}
\usepackage{amssymb}
\usepackage{xypic}
\usepackage[square]{natbib}
\setcitestyle{numbers}
\theoremstyle{plain}

\newtheorem{proposition}{Proposition}

\theoremstyle{definition}
\newtheorem{definition}{Definition}

\newtheorem{example}{Example}

\newtheorem{conjecture}{Conjecture}

\title[Topology of CSAS]{The Topology of Circular Synthetic Aperture Sonar Targets}
\author{Michael Robinson, Zander Memon}
\address{Mathematics and Statistics\\
American University\\
Washington, DC, USA}
\email{michaelr@american.edu}

\begin{document}

\begin{abstract}
  This report presents a connection between the physical acoustics of an object and the topology of the space of echoes resulting from a circular synthetic aperture sonar (CSAS) collection of that object.  A simple theoretical model is developed that yields a precise, yet qualitative, description of the space of echoes.  This theoretical model is validated in simulation and with experimental data from a laboratory sonar system.
\end{abstract}

\maketitle

\tableofcontents

\section{Introduction}

A circular synthetic aperture sonar (CSAS) collection can be formalized as a smooth function $u: S^1 \to \mathbb{C}^n$ from the circle $S^1$ (defining \emph{look angle}) to a vector of complex samples in $\mathbb{C}^n$.
Our previous work has revealed that understanding the topological and geometric structure of the \emph{signature space}, which is the image of $u$, can aid classification tasks \cite{sonarspace}.  Additionally, one can factor $u$ into the composition of two carefully chosen functions $u = U \circ \phi$ as described in \cite{Robinson_qplpf, Robinson_SampTA_2015}.  The intermediate space created, called the \emph{phase space}, provides a means to decouple the structure of the target from its environment \cite{Robinson_constrank}.  This report combines these two insights to determine how the physical acoustical structure of a CSAS collection impacts the signature and phase spaces.

\section{Theoretical model}

The topology of the signature space of a CSAS collection is highly constrained.  Intuitively, strong constraints on the signature space ensure that certain physical properties of a CSAS target will translate into its topological structure.

Because of Sard's theorem, critical points of an arbitrary smooth function $u$ are rare; they form a set of measure $0$.  In most physical situations, the critical points with extremal critical values will be isolated.  As an approximation, we will assume that all critical points are isolated.

\begin{proposition}
  \label{prop:csas_direct}
  Suppose that $u: S^1 \to \mathbb{C}^n$ is a smooth function with isolated critical points.  Then the image of $u$ is homeomorphic to a path-connected, compact, $1$-dimensional cell complex.  Moreover the image is homotopy equivalent to a wedge product of finitely many circles.
\end{proposition}
\begin{proof}
  Since $u$ is smooth, it is continuous.  Therefore, the image of $u$ inherits path-connectedness, compactness, and its maximum dimension from $S^1$.

  Because the critical points are assumed to be isolated, compactness implies that there are finitely many critical points.  We may therefore partition $S^1$ into the disjoint union of these critical points and finitely many connected open intervals.  This yields a cell complex structure for $S^1$.  The inverse function theorem ensures that $u$ is a local diffeomorphism (hence local homeomorphism) on each of the open intervals.  The image of such an open interval is therefore also an interval.  We therefore have a partition of the image of $u$ into a union of $0$-dimensional cells (the critical points) and $1$-dimensional cells (the intervals).

  It merely remains to show that each interval is attached to a critical point at each of its ends.  Consider an interval that is the image of $(\theta_1,\theta_2) \subseteq S^1$.  By construction, $\theta_1$ and $\theta_2$ are critical points of $u$, which may be the same point.  Therefore, $u((\theta_1,\theta_2))$ is attached to these two critical points.  In short, considering the cell complex decomposition of $S^1$, every critical point corresponds to a vertex of degree $2$.  

  Finally, every path-connected, $1$-dimensional cell complex with finitely many cells is homotopy equivalent to a wedge product of circles.  
\end{proof}

The image of $u$ may not be \emph{homeomorphic} to a wedge product of finitely many circles.  The cell complex structure for the image of $u$ might contain degree $1$ vertices, which can occur when the intervals on either side of such a vertex map to the same locus of points.

\begin{figure}
  \begin{center}
    \includegraphics[width=4in]{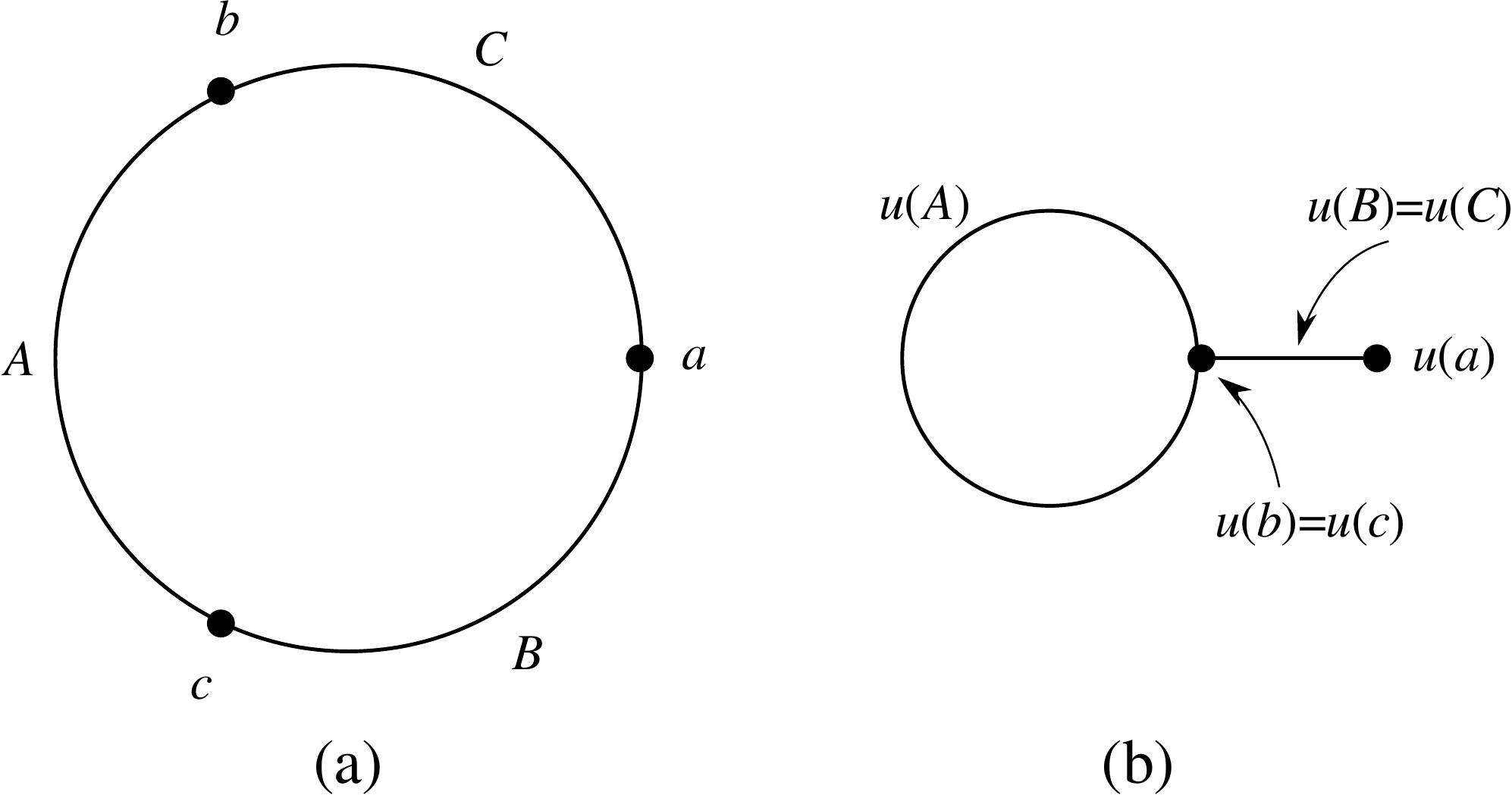}
    \caption{Formation of a flare in the image of a smooth map, as described in Example \ref{eg:collapse}.  (a) A cell complex decomposition of the circle.  (b) The image of a smooth map $u$, in which two of the $1$-dimensional cells are identified.}
    \label{fig:collapse}
  \end{center}
\end{figure}

\begin{example}
  \label{eg:collapse}
  Consider the case of a smooth map $u : S^1 \to \mathbb{R}^2$ with three critical points $a,b,c \in S^1$.  This induces a cell complex decomposition of the circle as shown in Figure \ref{fig:collapse}(a).  The map may be defined so that the critical point $a$ is flanked by two intervals with the same image as is shown in Figure \ref{fig:collapse}(b).  Notice that the degree of $a$ is $2$, but that the degree of $u(a)$ is $1$.  
\end{example}
    
A $1$-dimensional cell with a degree $1$ endpoint is called a \emph{flare}.  In CSAS data, flares correspond to specular reflections that are left/right symmetric.  With this insight, Proposition \ref{prop:csas_direct} is experimentally confirmed in Section \ref{sec:experimental}.

Our recent report \cite{Robinson_constrank} revealed that it is possible to decouple the sensor trajectory and signal propagation effects from the target structure using the idea of a \emph{quasiperiodic factorization}.

\begin{definition}\cite{Robinson_qplpf,Robinson_SampTA_2015}
  Recall that a \emph{submersion} is a constant rank map in which the Jacobian matrix is surjective as a linear map at every point in the domain.
  A \emph{quasiperiodic factorization of a smooth map $u: M \to N$} consists of a factorization $u = U \circ \phi$ in which $\phi : M \to C$ is a surjective submersion.  The space $C$ is called the \emph{phase space} for $u$.
\end{definition}

One can interpret the $\phi$ map of a quasiperiodic factorization as characterizing the trajectory and propagation effects.  The $C$ space represents the idealized space of echoes for that particular target.  The $U$ map represents the idealized structure of those echoes.

The surjective submersion assumption on $\phi$ ensures that the phase space $C$ is topologically similar to $M$.  This transfers all critical points of $u$ to critical points of $U$.  The surjective submersion assumption is a powerful inferential tool in the case of CSAS, because the domain $M$ has a very constrained and well-understood structure, namely it is a circle.

\begin{proposition}
  \label{prop:phase_structure}
  Suppose that $\phi: S^1 \to C$ is a surjective submersion of the circle.  Then $C$ is homeomorphic to a path-connected, compact, $1$-dimensional cell complex in which no vertex has degree $1$.
\end{proposition}

\begin{proof}
  Since $\phi$ is smooth, it is continuous.  $C$ inherits path-connectedness, compactness, and its dimension from $S^1$, because $\phi$ is also surjective.

  It remains to demonstrate degree $1$ vertices contradict the assumptions on $\phi$.  Suppose that $C$ does have a degree $1$ vertex.  This means that there is a subspace of $C$ that is homeomorphic to $[0,1)$, where the point $0$ maps to the degree $1$ vertex.  Notice that since the domain is $1$-dimensional, the Jacobian matrix of $\phi$ is just the usual derivative $\phi'$.  This derivative must evidently must take the value $\phi'(0)=0$ due to the fact that there are no points ``left of'' $0$ in $C=[0,1)$.  This means that $\phi$ has a critical point at $0$, which contradicts the fact that $\phi$ is a submersion and therefore cannot have critical points.
\end{proof}

In short, Proposition \ref{prop:phase_structure} asserts that the image of $\phi$ cannot contain flares.
In practice, one might expect to find apparent occurrences of flares in sampled data, due to the presence of aliasing and sampling error.

Using the algorithm described in \cite{Robinson_qplpf}, we may obtain a homeomorphic copy of $C$ from samples of a CSAS collection $u$.  The algorithm requires that we bundle angle-lagged copies of the domain, forming the vector
\begin{equation*}
  \Phi(\theta) := (u(\theta+\tau_1), u(\theta+\tau_2), u(\theta+\tau_3), \dotsc, u(\theta+\tau_N)),
\end{equation*}
where $\tau_1, \dotsc, \tau_N \in S^1$ are constants chosen arbitrarily from a certain dense subset of $(S^1)^N$.  (Intuitively, most choices for the constants $\tau_1, \dotsc, \tau_N$ are sufficient for the algorithm to work.)  For $N$ large enough, the Jacobian matrix of $\Phi$ becomes non singular for all $\theta$.  Therefore, a surjective $\phi$ can be obtained from $\Phi$ by restricting the codomain of $\Phi$ to its image.

\section{Qualitative structure of sonar targets}
\label{sec:qualitative}

As described in the previous section, we will formalize a CSAS collection as a smooth function $u: S^1 \to \mathbb{C}^n$.  The codomain $\mathbb{C}^n$ is naturally a Banach space, with norm $\|\cdot\|$.  Sonar signals are typically mean centered, so that the absence of an echo results in a signal $z \in \mathbb{C}^n$ with norm $\|z\| \approx 0$.  The receiver noise level sets the expected statistical properties of the norm of a signal in the absence of echoes.  These echoes which consist mostly of noise will appear as a Gaussian ``haze'' of echoes around the origin.

\begin{figure}
  \begin{center}
    \includegraphics[width=3.5in]{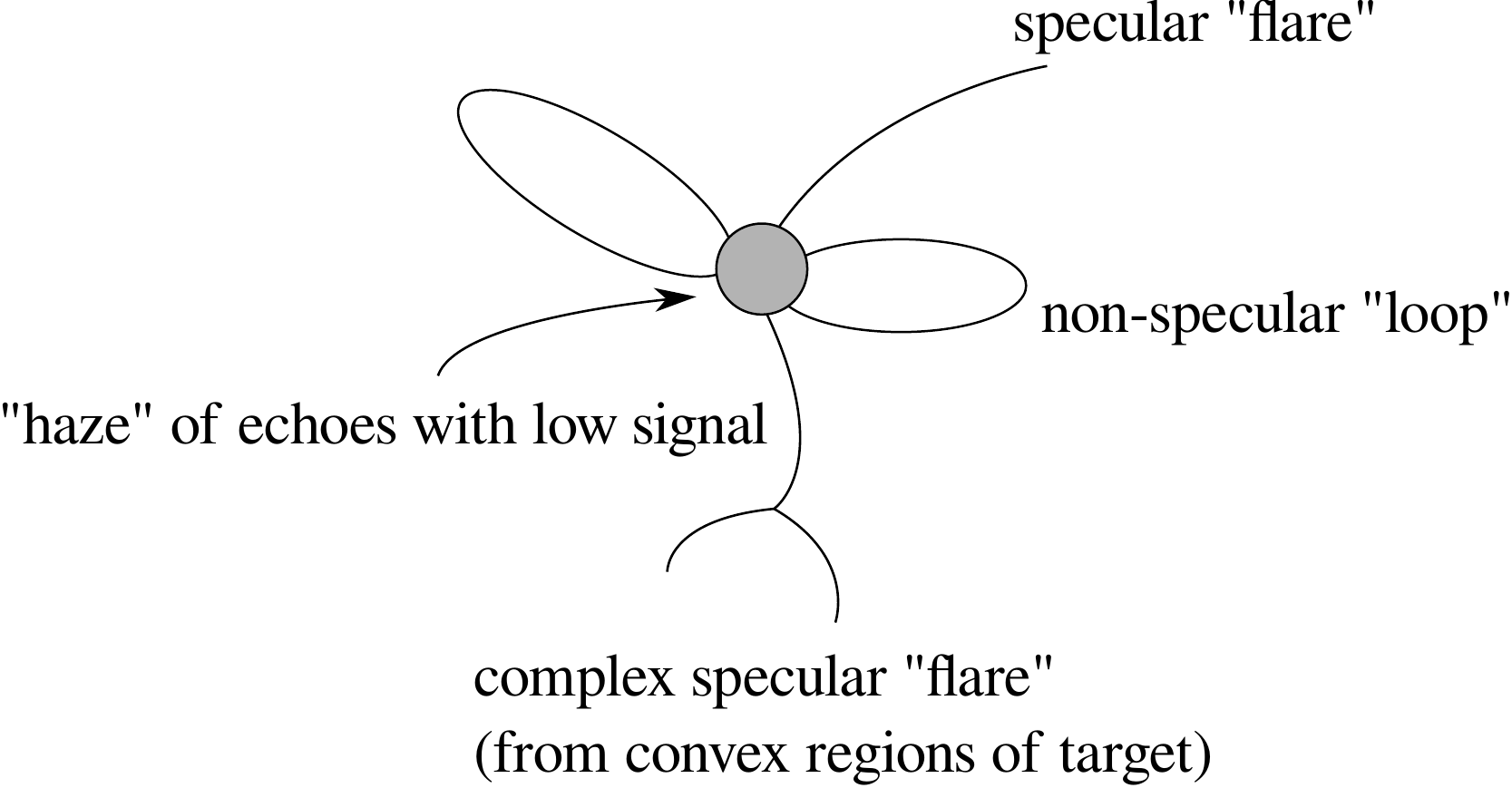}
    \caption{Schematic of the signature space for a typical CSAS target.}
    \label{fig:csas_schematic}
  \end{center}
\end{figure}

This consideration suggests that there are two kinds of geometric features that are expected in addition to the topological features.  The first is a preponderance of echoes near the origin, which correspond to angles for which minimal signal is received.  The second is that the local maximum norms of signals correspond to the length of the intervals from Proposition \ref{prop:csas_direct}.  These correspond to the acoustic cross-sections of specular reflections.  Informally, this suggests that the signature space will consist of a dense ball of echoes, to which specular ``flares'' and non-specular ``loops'' are attached, as shown in Figure \ref{fig:csas_schematic}.

In short, we expect that the number of prominent specular reflections to correspond to excursions of the signature space away from the origin.  If a specular reflection is left/right symmetric in angle it will appear as a ``flare,'' whose maximum norm corresponds to the sonar cross-section of the reflection.  On the other hand, if there is an asymmetry, it will appear as a ``loop,'' again whose maximum norm corresponds to the sonar cross-section of the reflection.

The ``flares'' do not correspond to topological features, and so we do not expect to see them reflected in the persistent homology of the signature space.  The ``loops'' do correspond to topological features, and so should be expected to be visible in the $H_1$ persistence diagram.  Their birth time tends to $0$ as the sample rate (samples per degree of look angle) increases.  Their death time corresponds to the semi-minor axis of the loop, which is governed by both the sonar cross-section of the reflection and the amount of left/right asymmetry it exhibits.

If we instead consider the phase space $C$ for the universal quasiperiodic factorization obtained using the algorithm described in \cite{Robinson_qplpf}, so that $u = U \circ \Phi$ and $\Phi: S^1 \to C$, then each ``flare'' is transformed into a loop via Proposition \ref{prop:phase_structure}.  In this case, the persistent homology can detect all prominent specular reflections, assuming their geometry is favorable.  In particular, instead of using the geometry obtained by pulling back the metric of $\mathbb{C}^n$ to $C$, we should use the natural norm on $C \subset \mathbb{C}^{nN}$, where $N$ is the number of angle-lagged copies in use.

In this case, the persistent $H_1(C)$ generators correspond to all of the prominent specular reflections.  Again, the birth time is dominated by the sample rate, and decreases to $0$ as the sample rate is increased.  For finite sample rates, the birth time involves the sonar cross section of each specular reflection, its beam width, and the sample rate.  The death time corresponds to the sonar cross section of the specular reflection as before, but also to its beam width.

\begin{conjecture}
  The persistence diagram of $H_1(C)$ can be deduced by the coefficients of $u$ when decomposed in an appropriate wavelet basis.
\end{conjecture}

\section{Simulated data}
\label{sec:simulated}

Consider a target consisting of seven (7) point scatterers of equal reflectivity placed on a circle.  Suppose further that the scatterers are divided into two groups, one group consisting of 2 scatterers and one consisting of 5 scatterers.  Each group of scatterers is distributed equally spaced in angle around the center of the circle, but the two groups are not aligned with each other.

\begin{figure}
  \begin{center}
    \includegraphics[width=4in]{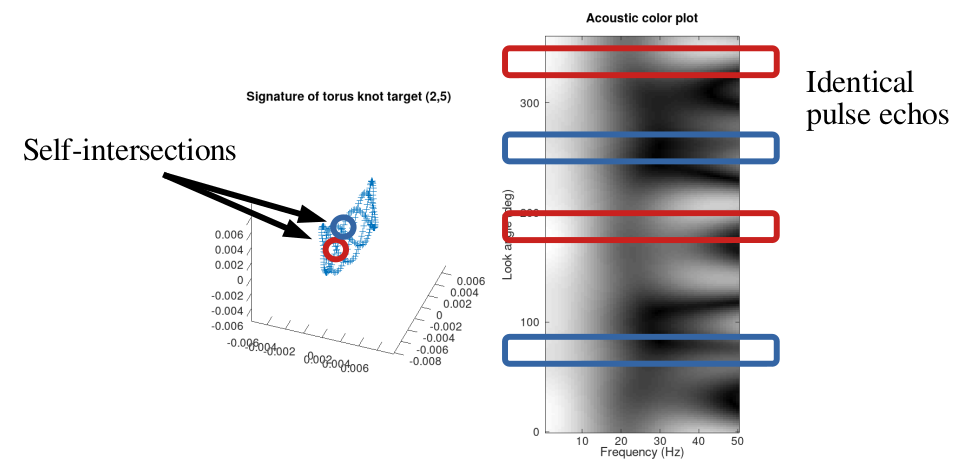}
    \caption{The signature space for a CSAS collection of a target consisting of seven point scatterers.}
    \label{fig:csas_direct}
  \end{center}
\end{figure}

If we illuminate the target with a CSAS system, the phase space for this set of scatterers is a torus knot \cite[Prop. 3]{Robinson_constrank}.  A visualization of the signature space (not the phase space), once projected to $3$ dimensions using principal components analysis, is shown in Figure \ref{fig:csas_direct}.  The signature space contains notable self-intersections.  These self-intersections are the places where the wedge products mentioned in Proposition \ref{prop:csas_direct} are performed.  No intervals are apparent in the plot.

\begin{figure}
  \begin{center}
    \includegraphics[width=4in]{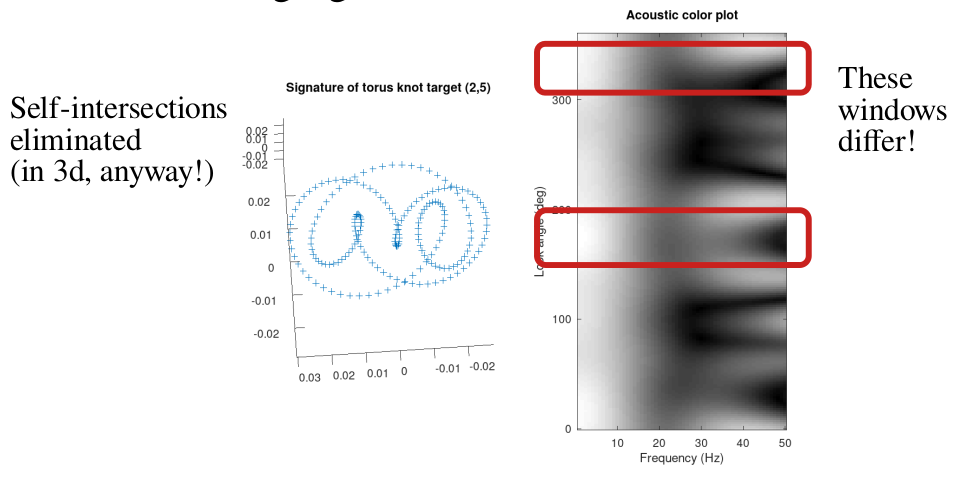}
    \caption{The phase space for a CSAS collection of a target consisting of seven point scatterers, in which sliding windows in angle are used.}
    \label{fig:csas_sliding}
  \end{center}
\end{figure}

The fact that there are self-intersections is apparently in contradiction with the torus knot structure, since torus knots have no self-intersections.
The algorithm described in \cite{Robinson_qplpf} instructs us to use angle-lagged copies of the data to resolve spurious self-intersections.  This produces the phase space shown in Figure \ref{fig:csas_sliding}, which exhibits no self-intersections in $3$-dimensions. We conclude that the algorithm in \cite{Robinson_qplpf} is indeed successful on the simulated data.

\section{Experimental validation}
\label{sec:experimental}

We considered CSAS collections of three household objects (a copper pipe, a styrofoam cup, and a coke bottle) in two configurations (open ends or capped ends) each using ARL/PSU's AirSAS system \cite{cowenairsas} as shown in Figure \ref{fig:airsas_targets}.

\begin{figure}
  \begin{center}
    \includegraphics[width=5in]{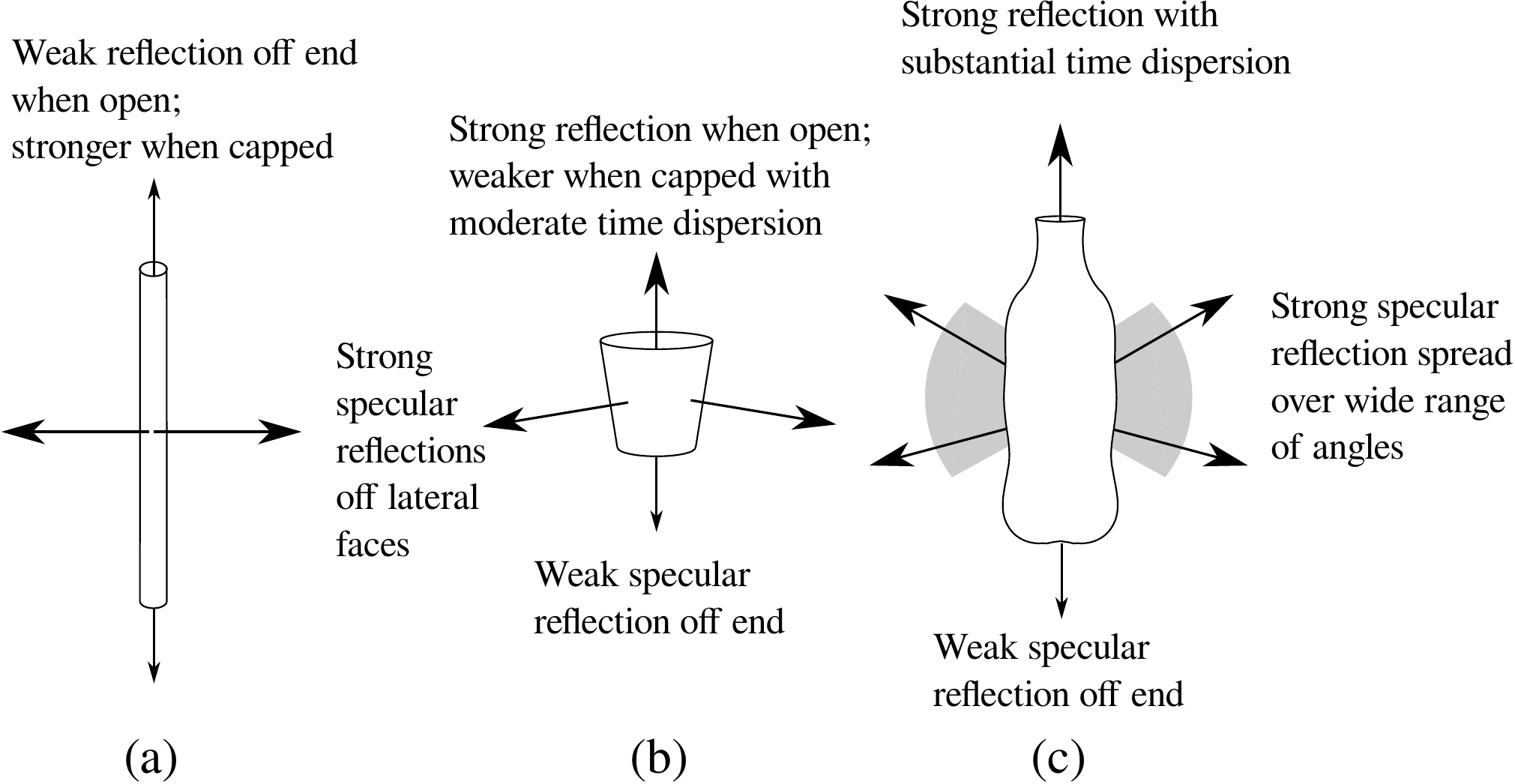}
    \caption{Schematic diagrams of the three targets used in the experiment: (a) a copper pipe, (b) a styrofoam cup, and (c) a coke bottle.}
    \label{fig:airsas_targets}
  \end{center}
\end{figure}

The objects were placed at the center of a rotating turntable and were imaged from a fixed sonar sensor located a few meters away from the center.
The collected data consist of a $2$ dimensional array of data, in which the rows correspond to angles and the columns correspond to range.
The array for each target was the same shape, and contains one sample per degree (a total of $360$ rows) and $1000$ range samples (columns).

The three targets have sonar cross sections that depend on look angle in distinctly different ways, as shown schematically in Figure \ref{fig:airsas_targets}.  The copper pipe and styrofoam cup will show strong specular reflections with narrow beam width from their lateral faces, because these are large flat surfaces.  The lateral faces of the coke bottle can also be expected to show a strong specular reflection, though its beam width will be much larger due to curvature.  When the ends of each object are closed, they present a sub-wavelength aperture, and so should not be expected to have a large sonar cross section.  When the ends are open, the objects can resonate and reradiate the sound energy, so large non-specular reflections from the open ends are possible.  However, since the reradiation process takes time, one should expect time/range dispersion on these echoes.  From the standpoint of the space of echoes, this means that echoes corresponding to the lateral faces will be located far from those from an open end, even if their cross sections are similar.

We can visualize the space of echoes using \emph{principal components analysis} (PCA).  For each target in each configuration, we computed two PCA plots: one for the unmodified data (the signature space for $u$ in $\mathbb{C}^{1000}$), and one with a sliding window with $N=3$ angle-lagged offsets (the phase space of $u$, or the image of $\Phi$, in $\mathbb{C}^{3000}$).  Many choices of the offsets will work.  We arbitrarily chose to use $\tau_1 = 0^\circ$, $\tau_2 = 4^\circ$, and $\tau_3 = 25^\circ$ offsets for the results shown here.  As suggested in Proposition \ref{prop:phase_structure} and later in Section \ref{sec:qualitative}, the sliding window results should ensure that all prominent specular reflections become topological ``loops.'' To verify that assertion, we also display the persistence diagrams for $H_0$ and $H_1$ for both the signature and phase spaces.

\subsection{Copper pipe with capped ends}

\begin{figure}
  \begin{center}
    \includegraphics[width=3in]{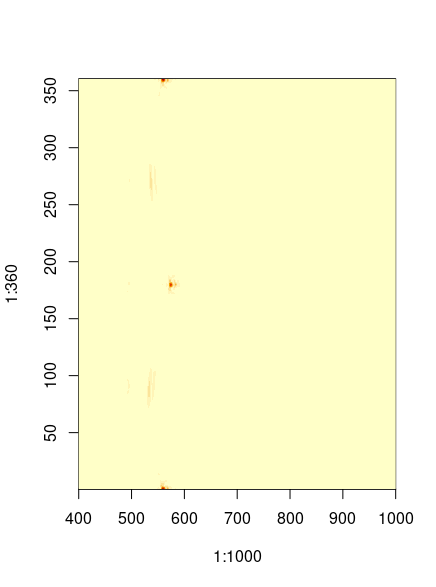}
    \caption{The AirSAS data for copper pipe with capped ends: the horizontal axis is range, the vertical axis is look angle.}
    \label{fig:pipe_capped_raw_data}
  \end{center}
\end{figure}

Figure \ref{fig:pipe_capped_raw_data} shows strong, highly focused specular reflections off the lateral faces at $0^\circ = 360^\circ$ and at $180^\circ$.  The focused nature of these echoes indicates that the lateral faces of the pipe are flat and smooth.
There are weak, angularly dispersed echoes around $90^\circ$ and $270^\circ$, which are likely the ends of the pipe. 

\begin{figure}
  \begin{center}
    \includegraphics[width=5in]{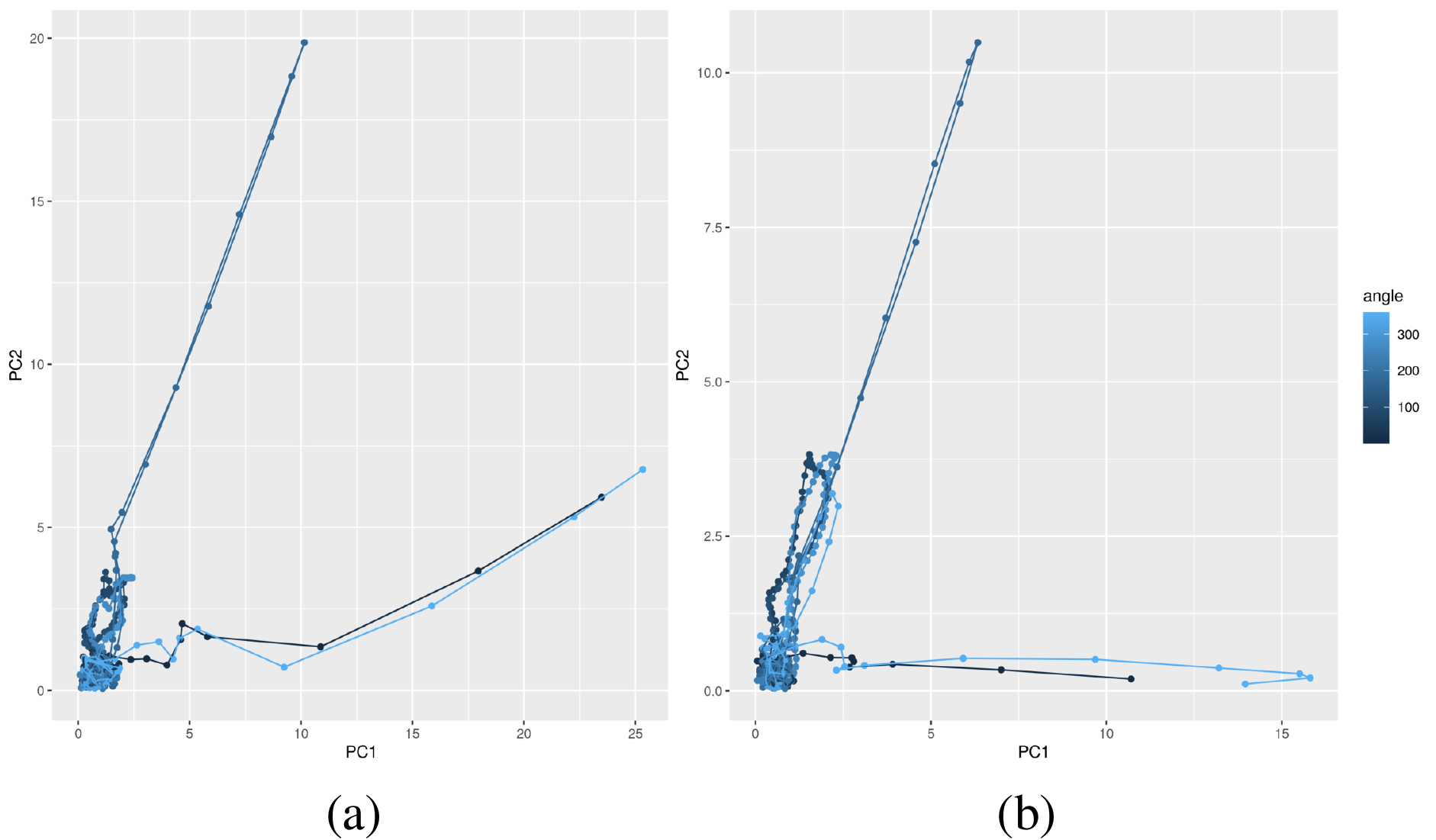}
    \caption{PCA plot of copper pipe with capped ends signature space (a) and phase space (b).}
    \label{fig:pipe_capped_pca_data}
  \end{center}
\end{figure}

In Figure \ref{fig:pipe_capped_pca_data}(a), the strong specular reflection at $0^\circ = 360^\circ$ corresponds to a loop.  The strong specular reflection from the opposite face at $180^\circ$ corresponds to a flare.  These features appear to remain in Figure \ref{fig:pipe_capped_pca_data}(b), though the weaker reflections are now potentially visible as loops.

\begin{figure}
  \begin{center}
    \includegraphics[width=5in]{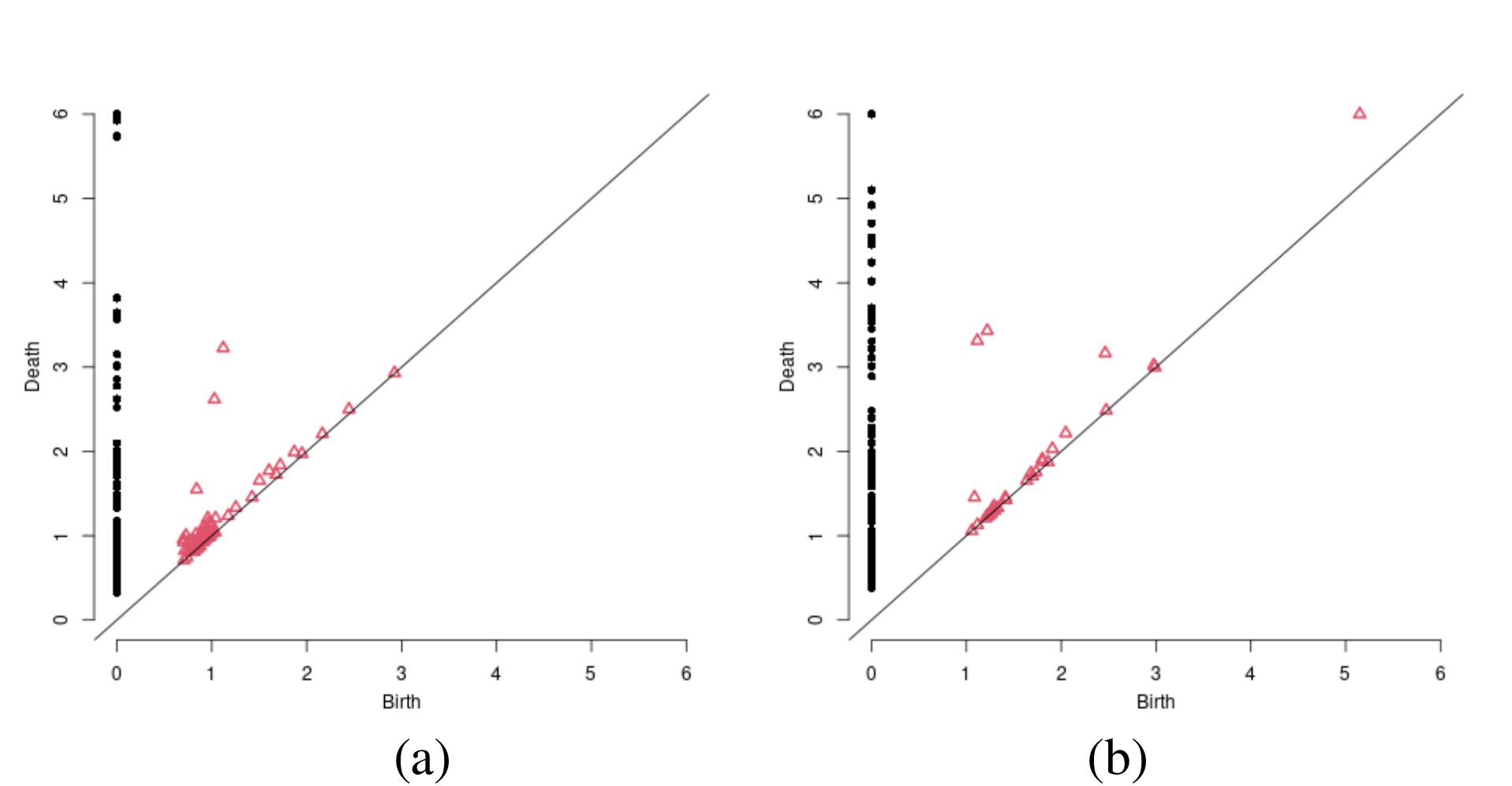}
    \caption{Persistence diagrams of copper pipe with capped ends signature space (a) and phase space (b).  Black points form the diagram from $H_0$; red points form the diagram for $H_1$.}
    \label{fig:pipe_capped_persistence}
  \end{center}
\end{figure}

Figure \ref{fig:pipe_capped_persistence} shows the persistence diagrams for the signature and phase spaces for the copper pipe with capped ends.  The distance from the diagonal corresponds to the \emph{lifetime} of the topological features.  Notice that the lifetime is precisely the difference between death and birth times shown on the persistence diagram.  Longer lifetimes are associated with more robust features.  In Figure \ref{fig:pipe_capped_persistence}(a), there is one clear long lifetime feature for $H_1$ and one shorter (yet still long) lifetime feature.  These two features likely correspond to the specular reflections from the lateral faces.  The situation is much clearer in Figure \ref{fig:pipe_capped_persistence}(b), where there are two clear long-lifetime features (corresponding to the lateral face specular reflections). Two features are also visible away from the diagonal that likely correspond to the reflections off the ends.  This is a clear confirmation of Proposition \ref{prop:phase_structure}.

\subsection{Copper pipe with open ends}

\begin{figure}
  \begin{center}
    \includegraphics[width=3in]{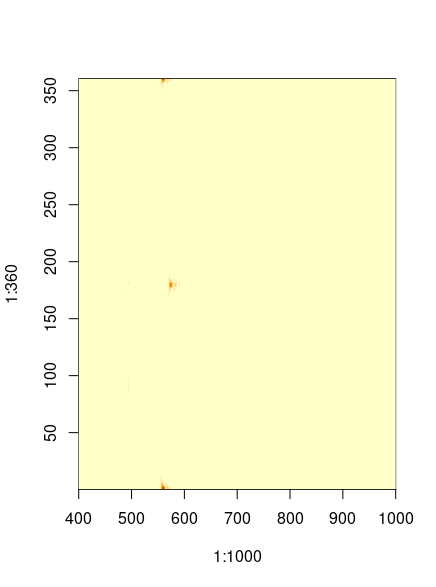}
    \caption{The AirSAS data for copper pipe with open ends: the horizontal axis is range, the vertical axis is look angle.}
    \label{fig:pipe_open_raw_data}
  \end{center}
\end{figure}

In Figure \ref{fig:pipe_open_raw_data}, there are three main specular reflections visible at $0^\circ = 360^\circ$ and $180^\circ$.
However, the strength of the echoes is less than the strength of the comparable echoes in the capped copper pipe. 

We do not know why there are apparently no reflections whatsoever from the ends of the pipe. This may be due to creeping wave losses, or simply the fact that the physical cross-section of the end of an open pipe is acoustically very small.

\begin{figure}
  \begin{center}
    \includegraphics[width=5in]{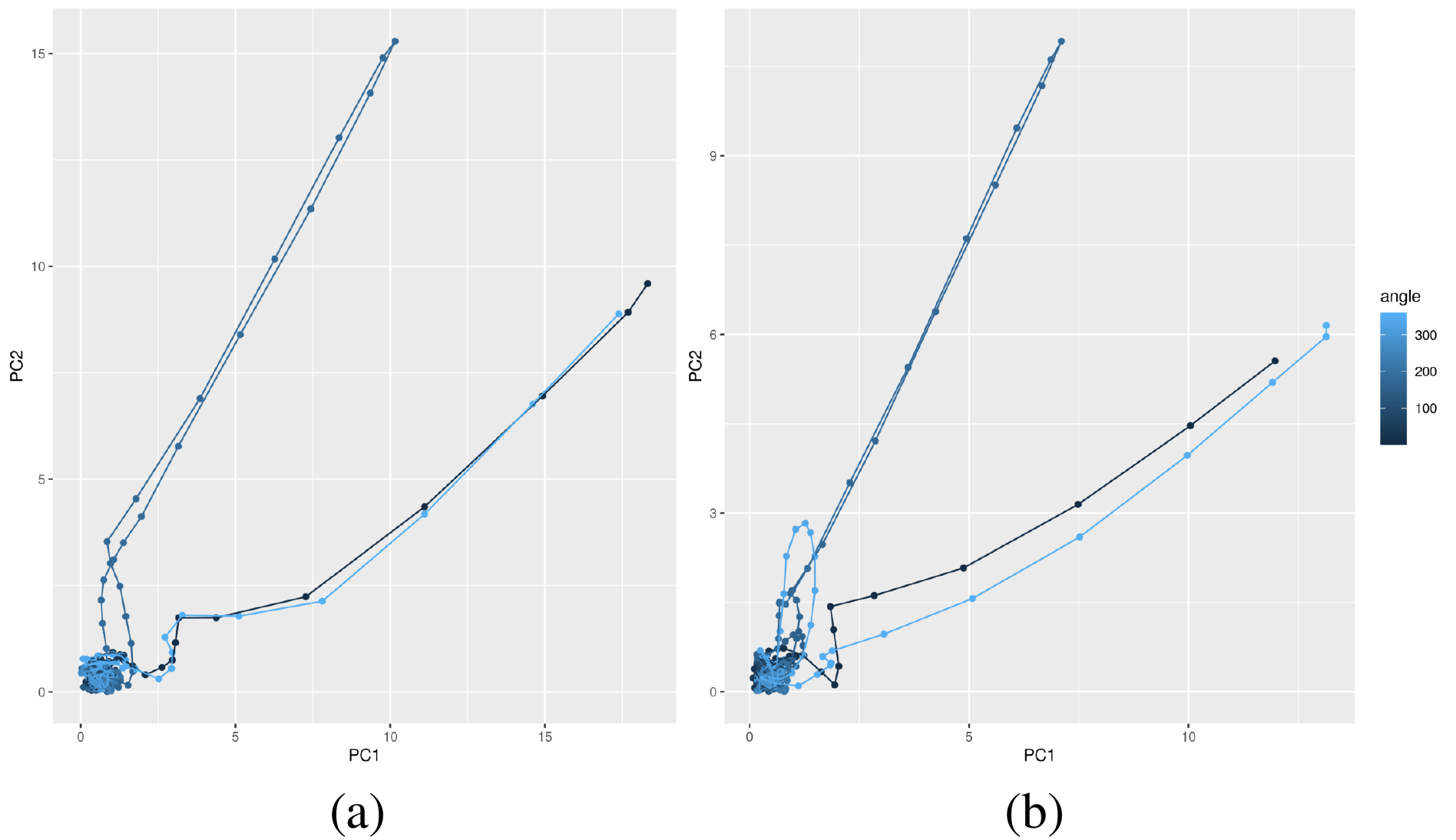}
    \caption{PCA plot of copper pipe with open ends signature space (a) and phase space (b).}
    \label{fig:pipe_open_pca_data}
  \end{center}
\end{figure}

In Figure \ref{fig:pipe_open_pca_data}(a), two flares correspond to the strong specular reflections from the lateral faces.  This structure is also visible in Figure \ref{fig:pipe_open_pca_data}(b), though as explained in the previous sections, the flares are actually loops.  This is not apparent from Figure \ref{fig:pipe_open_pca_data}(b), though it is quite clear if we compare the persistence diagrams in Figure \ref{fig:pipe_open_persistence}.

\begin{figure}
  \begin{center}
    \includegraphics[width=5in]{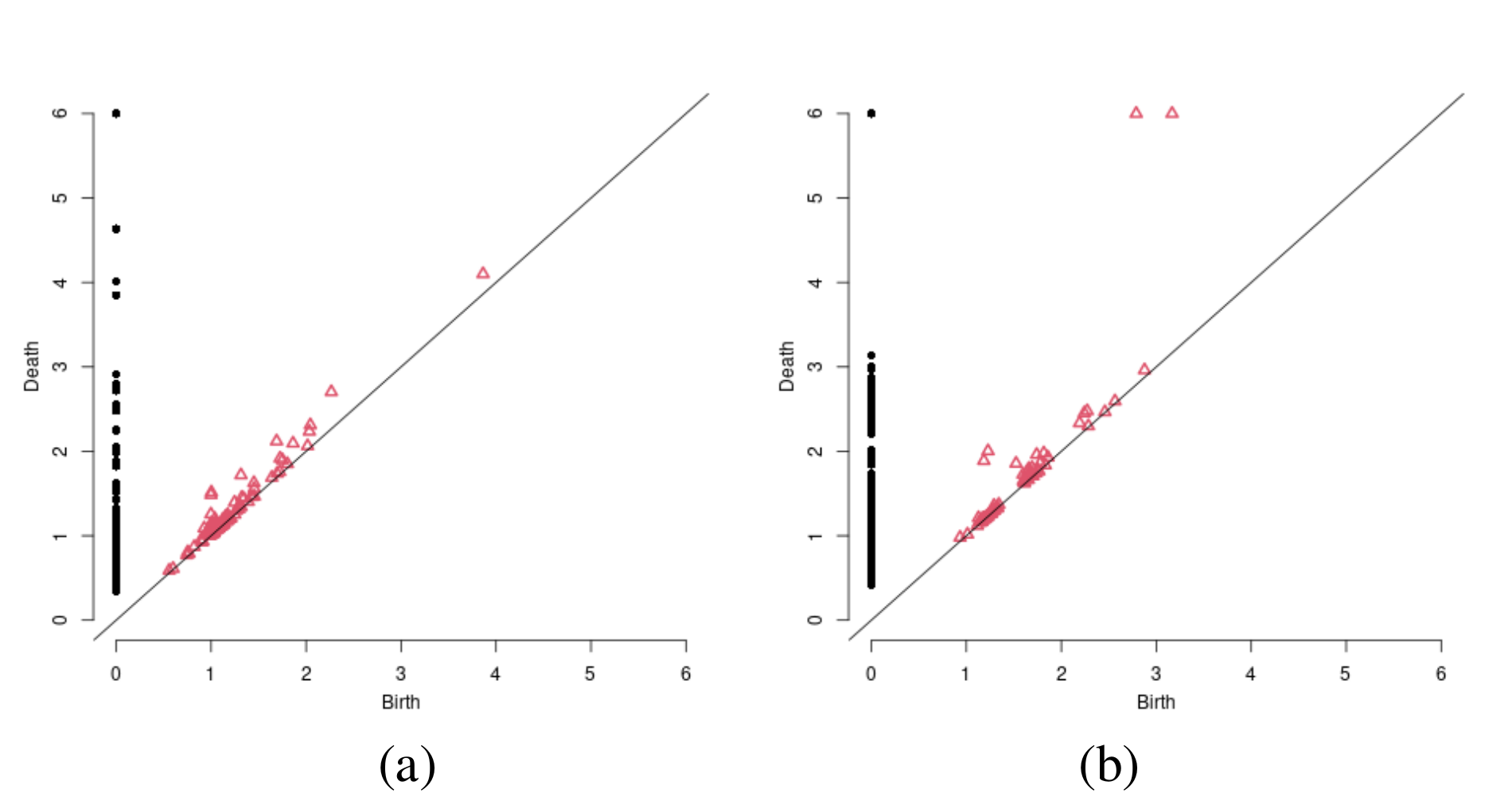}
    \caption{Persistence diagrams of copper pipe with open ends signature space (a) and phase space (b).  Black points form the diagram from $H_0$; red points form the diagram for $H_1$.}
    \label{fig:pipe_open_persistence}
  \end{center}
\end{figure}

Notice that both loops appear quite similar in both frames of Figure \ref{fig:pipe_open_pca_data}.  The differences in cross section may correspond to a slight misalignment of the center of the pipe and the center of the turntable.  This misalignment likely accounts for the slight difference in location of the topological features visible in Figure \ref{fig:pipe_open_persistence}(b).

\clearpage
\subsection{Styrofoam cup with open lid}

\begin{figure}
  \begin{center}
    \includegraphics[width=3in]{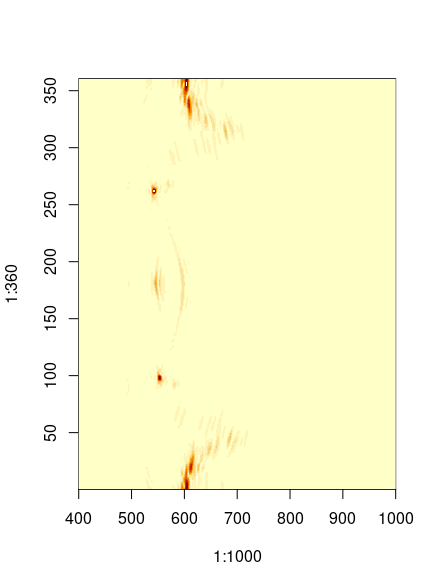}
    \caption{The AirSAS data for styrofoam cup with open lid: the horizontal axis is range samples, the vertical axis is look angle in degrees.}
    \label{fig:cup_open_raw_data}
  \end{center}
\end{figure}

There are three groups of echoes that are evident in Figure \ref{fig:cup_open_raw_data}.
The largest echoes occur approximately centered at $0^\circ=360^\circ$.
These result from the sound entering the mouth of the cup.
The echoes from $0^\circ = 360^\circ$
have a strongly time-dispersed signal that continues to reverberate after the initial echo.
This dispersion results from the pulses bouncing back and forth inside the cup. 

There are two echoes at approximately $100^\circ$ and $260^\circ$, which correspond to the specular reflections from the lateral faces of the cup.  These are strong signals that are highly focused in angle. Notice that the angle difference between these reflections is $160^\circ$, which suggests that the cup has the shape of a truncated cone.

Finally there is a pair of angularly dispersed echoes separated in time centered at $180^\circ$.  The first is at range sample $500$ and the second is at range sample $600$.  There are caused by two distinct mechanisms.  First, the waves bounce off the bottom of the cup, yielding the feature at range sample $500$. Approximately $100$ range samples later, waves return after diffracting off the mouth of the cup.  These diffracted waves have a larger angular dispersion, which is easily confirmed by referring to Figure \ref{fig:cup_open_raw_data}.

\begin{figure}
  \begin{center}
    \includegraphics[width=5in]{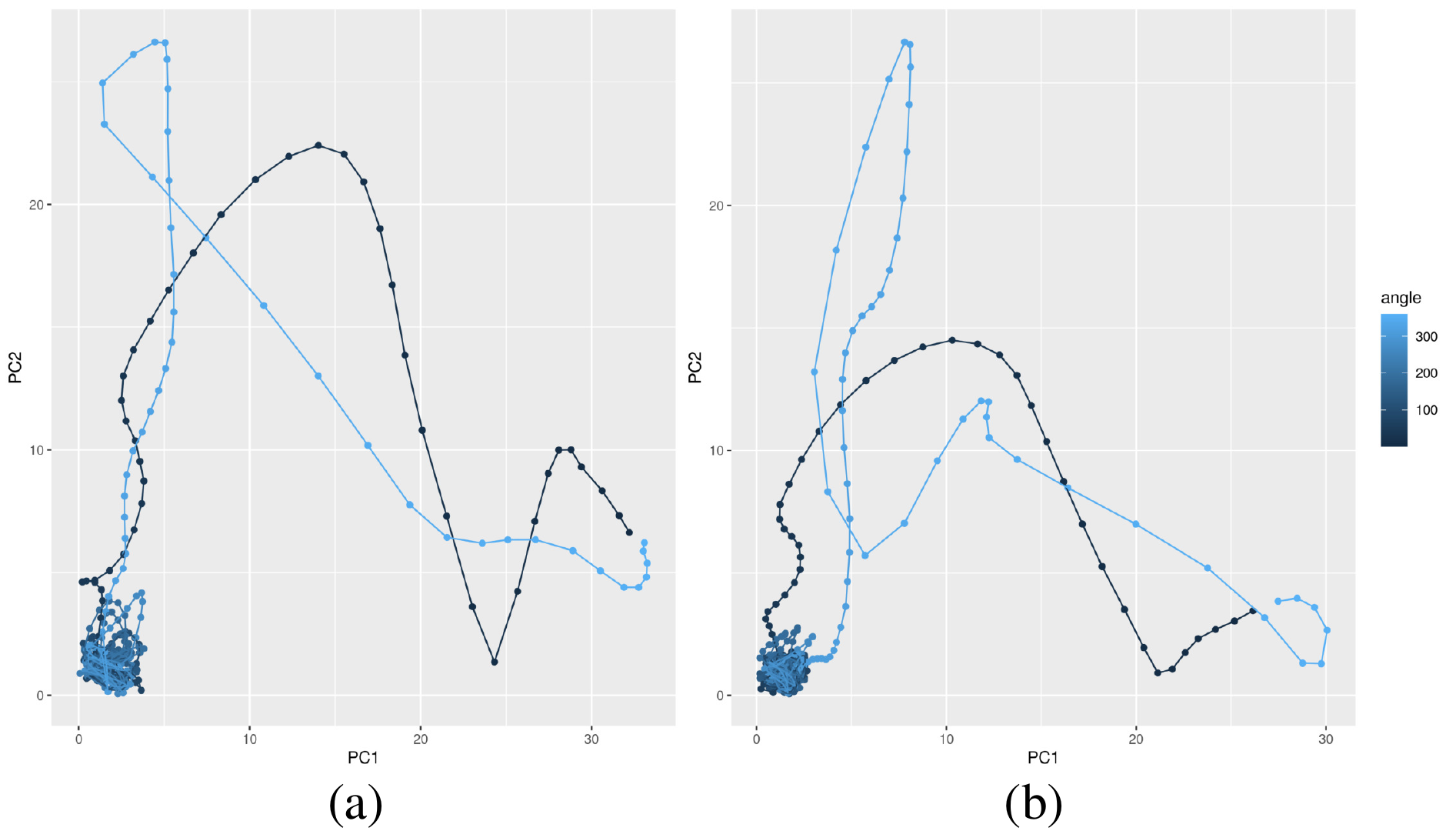}
    \caption{PCA plot of styrofoam cup with open lid signature space (a) and phase space (b).}
    \label{fig:cup_open_pca_data}
  \end{center}
\end{figure}

Figure \ref{fig:cup_open_pca_data}(a) shows the PCA plot of the signature space.
Due to the angular and time dispersion, the large reflection from the mouth of the cup at $0^\circ = 360^\circ$ dominates the plot.
Using sliding windows resolves the other echoes somewhat more clearly, resulting in Figure \ref{fig:cup_open_pca_data}(b).

In Figure \ref{fig:cup_open_pca_data}(a), there are smaller loops protruding from the low-signal ``haze'' around the origin. These loops
correspond to the reflections from $100^\circ$ and $260^\circ$ degrees where the pulses
bounce off the lateral faces of the cup. There are points that correspond
to the reflections from the pulses that reflect off the bottom of the cup at 
$180^\circ$ degrees.   For reasons that are unclear, these features are less visible in Figure \ref{fig:cup_open_pca_data}(b).

\begin{figure}
  \begin{center}
    \includegraphics[width=5in]{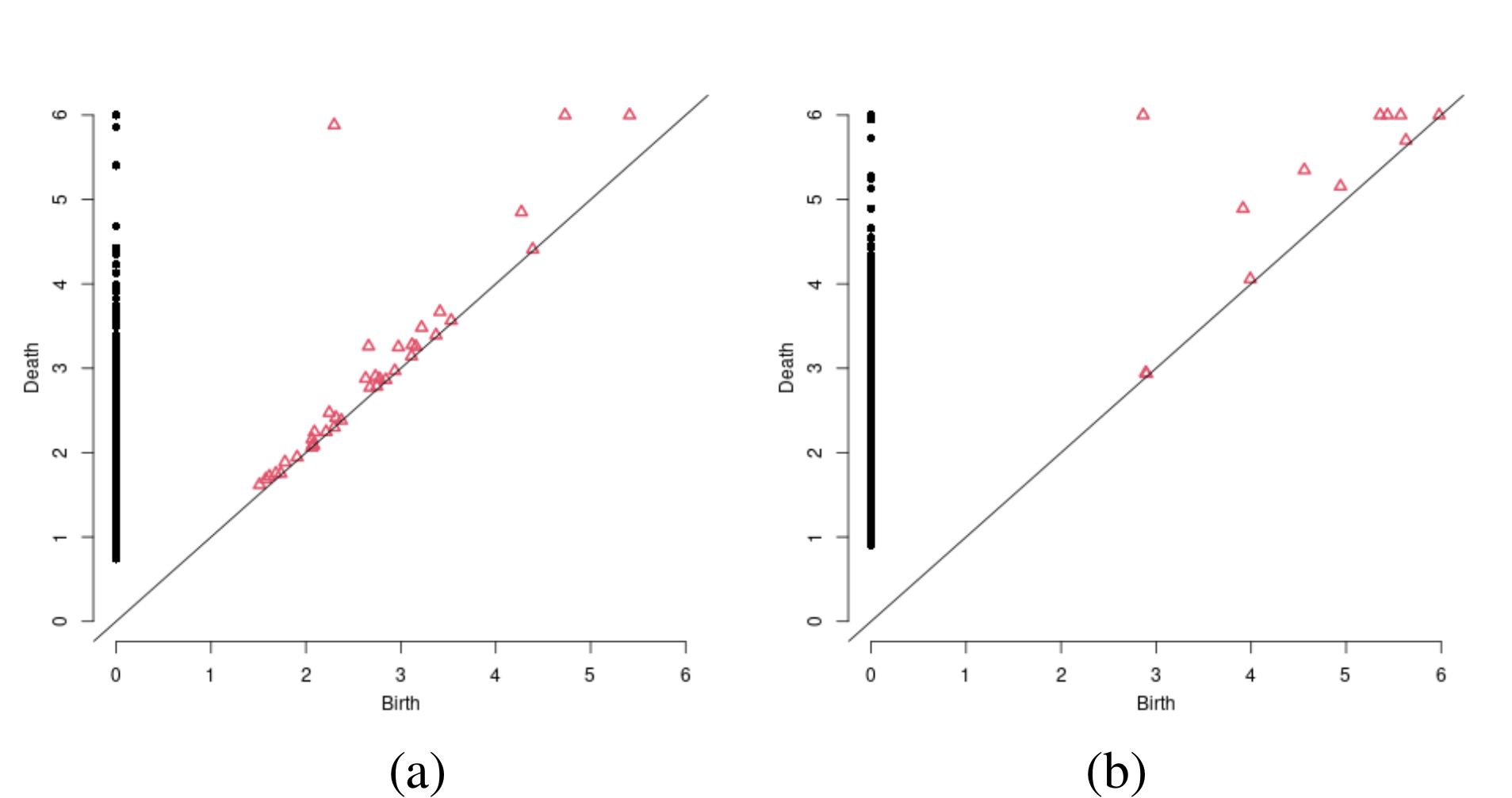}
    \caption{Persistence diagrams of styrofoam cup with open lid signature space (a) and phase space (b).  Black points form the diagram from $H_0$; red points form the diagram for $H_1$.}
    \label{fig:cup_open_persistence}
  \end{center}
\end{figure}

The persistence diagrams for the signature space (Figure \ref{fig:cup_open_persistence}(a)) and the phase space (Figure \ref{fig:cup_open_persistence}(b)) are fairly similar.  Both are dominated by the strong, complex reflection from the mouth of the cup, with several smaller loops also visible.

\subsection{Styrofoam cup with capped lid}

\begin{figure}
  \begin{center}
    \includegraphics[width=3in]{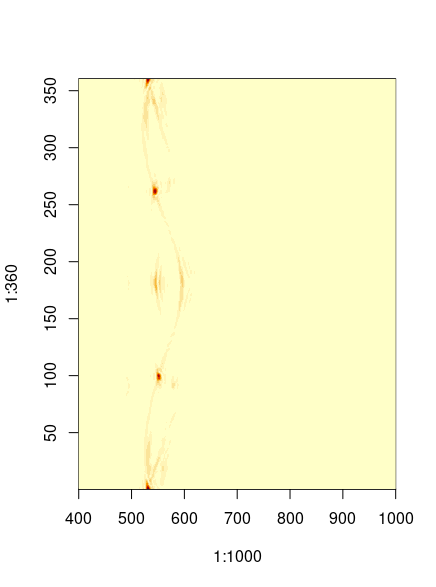}
    \caption{The AirSAS data for styrofoam cup with capped lid: the horizontal axis is range samples, the vertical axis is look angle in degrees.}
    \label{fig:cup_capped_raw_data}
  \end{center}
\end{figure}

We note that Figure \ref{fig:cup_capped_raw_data} is similar to Figure \ref{fig:cup_open_raw_data},
indicating that there is a similarity between the echoes of the capped and uncapped cups.
The strength of the echoes at $0^\circ=360^\circ$ are greatly reduced in Figure \ref{fig:cup_capped_raw_data},
because the sound does not enter the mouth of the cup.
The  initial echo from the entrance of the mouth of the cup is still present, but the echoes after this point (to the right of Figure \ref{fig:cup_capped_raw_data} are 
much more faint. This may indicate that the larger echoes from the uncapped 
cup are a result of the reverberations of the echoes once they enter the mouth of the cup. 

Once again, Figure \ref{fig:cup_capped_raw_data} exhibits focused specular reflections at
$100^\circ$ and $260^\circ$ from where the sound hits the lateral faces of the cup. 

Finally the pair of time-separated echoes at $180^\circ$ are still present, where the sound bounces off the 
bottom of the cup and diffracts from its mouth. These echoes are very close in strength to the echoes at $180^\circ$ 
from the uncapped cup.  From this angle, the uncapped and capped cups appear acoustically similar.

\begin{figure}
  \begin{center}
    \includegraphics[width=5in]{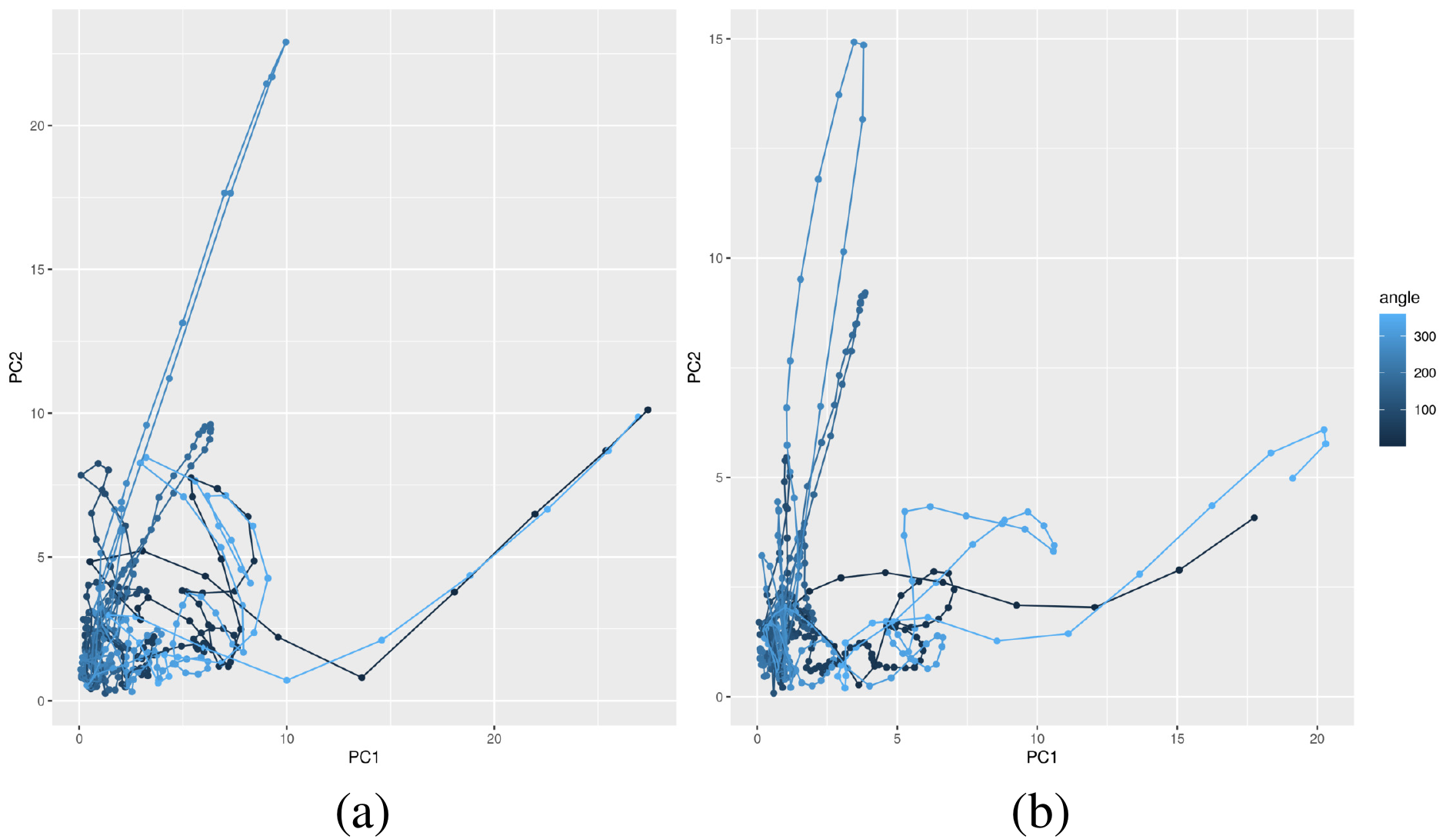}
    \caption{PCA plot of styrofoam cup with capped lid signature space (a) and phase space (b).}
    \label{fig:cup_capped_pca_data}
  \end{center}
\end{figure}

Figure \ref{fig:cup_capped_pca_data}(a) shows the PCA plot of the signature space of the capped cup.
The distinct loop at $0^\circ = 360^\circ$ from the open cup is still present.
Since the reflection from the mouth of the cup is smaller in cross section,
the topological structure of the weaker echoes is more visible for the capped cup in Figure \ref{fig:cup_capped_pca_data} than for the open cup in Figure \ref{fig:cup_open_pca_data}.

There is a loop that appears to 
correspond to the $260^\circ$ reflection. There is also 
a large, distinct, vertical loop that corresponds to the $100^\circ$ reflection.
Finally there is a small loop that corresponds to the $180^\circ$ reflection from the bottom of the cup. 
These loops are more clearly visible in the phase space shown in Figure \ref{fig:cup_capped_pca_data}(b),
which is a reflection of Proposition \ref{prop:phase_structure}.

\begin{figure}
  \begin{center}
    \includegraphics[width=5in]{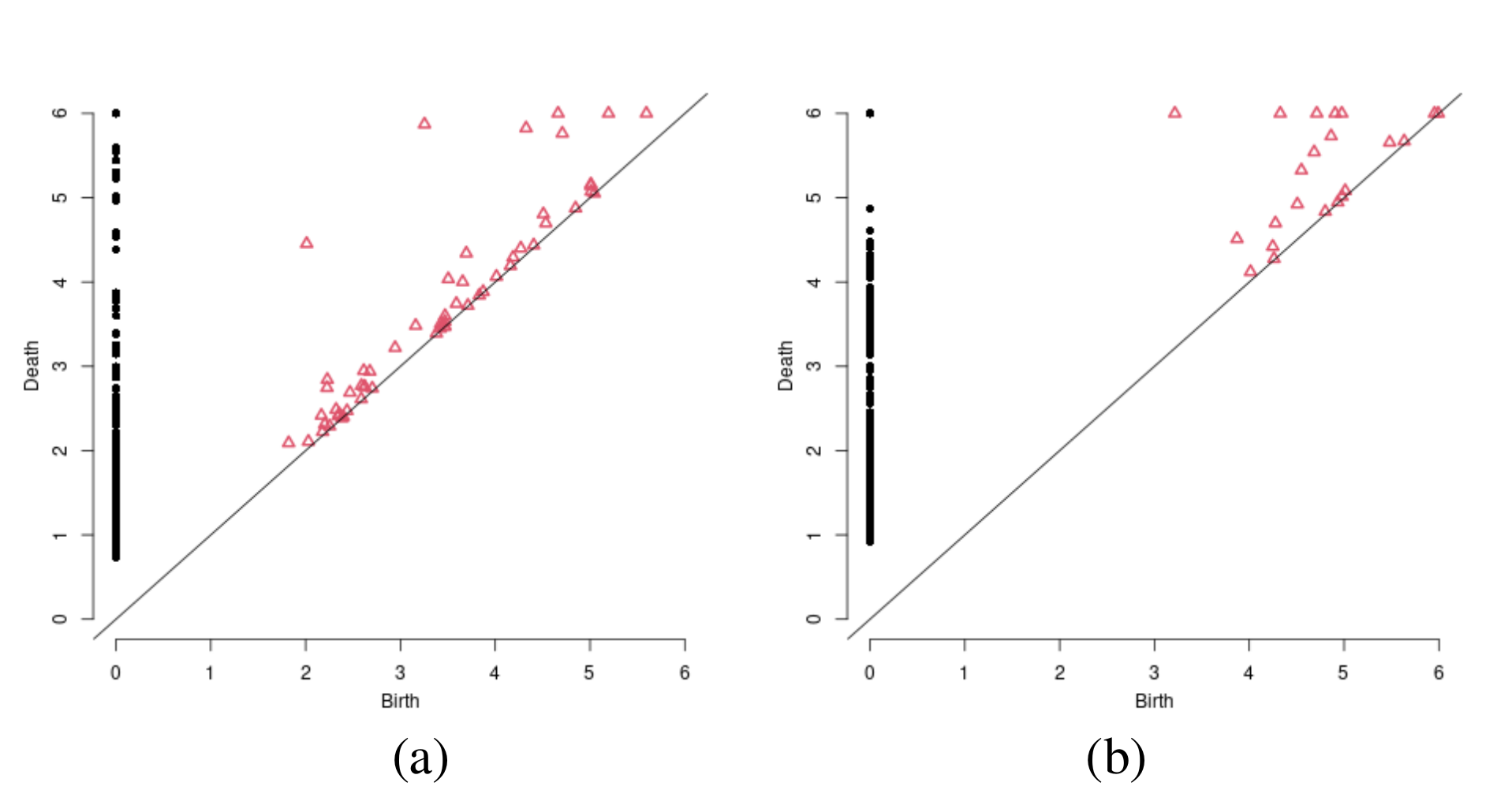}
    \caption{Persistence diagrams of styrofoam cup with capped lid signature space (a) and phase space (b).  Black points form the diagram from $H_0$; red points form the diagram for $H_1$.}
    \label{fig:cup_capped_persistence}
  \end{center}
\end{figure}

The persistence diagrams for the capped cup are shown in Figure \ref{fig:cup_capped_persistence}.  There is considerable complexity visible and interpretation of these plots is difficult.

\clearpage
\subsection{Coke bottle with open lid}

\begin{figure}
  \begin{center}
    \includegraphics[width=3in]{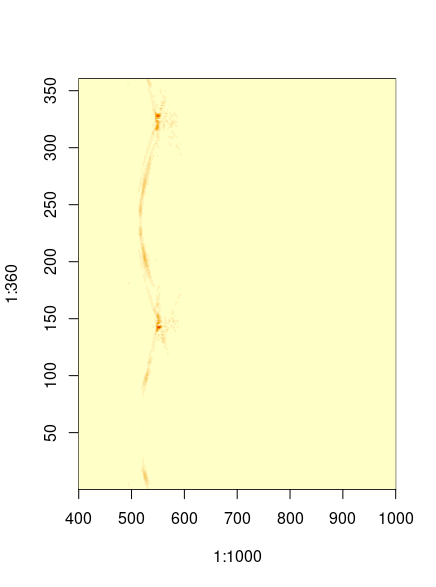}
    \caption{The AirSAS data for coke bottle with open lid: the horizontal axis is range samples, the vertical axis is look angle in degrees.}
    \label{fig:coke_open_raw_data}
  \end{center}
\end{figure}

Figure \ref{fig:coke_open_raw_data} shows the data for the coke bottle with an open lid.
The axis of symmetry for the coke bottle is not aligned with $0^\circ$ or $180^\circ$, as was the case for the previous targets.

Two main echoes are visible at approximately $150^\circ$ and $330^\circ$. 
These echoes result from when the sound hits the sides of the bottle.
The strength of these echoes is much smaller than the comparable echoes
from the cup in either configuration.
This may be due to a difference in material composition of the bottle versus the cup.

There are some angularly dispersed reflections extending between $150^\circ$ and $330^\circ$.
These are likely due to the curvature of the sides of the bottle.
Comparison between Figure \ref{fig:coke_open_raw_data} and Figure \ref{fig:coke_capped_raw_data} suggests that the main difference occurs at an angle of $50^\circ$, which is likely the location of the cap.
Therefore, we conclude that the bottom of the bottle is around $50^\circ + 180^\circ = 230^\circ$, and does not result in a substantial specular reflection.

\begin{figure}
  \begin{center}
    \includegraphics[width=5in]{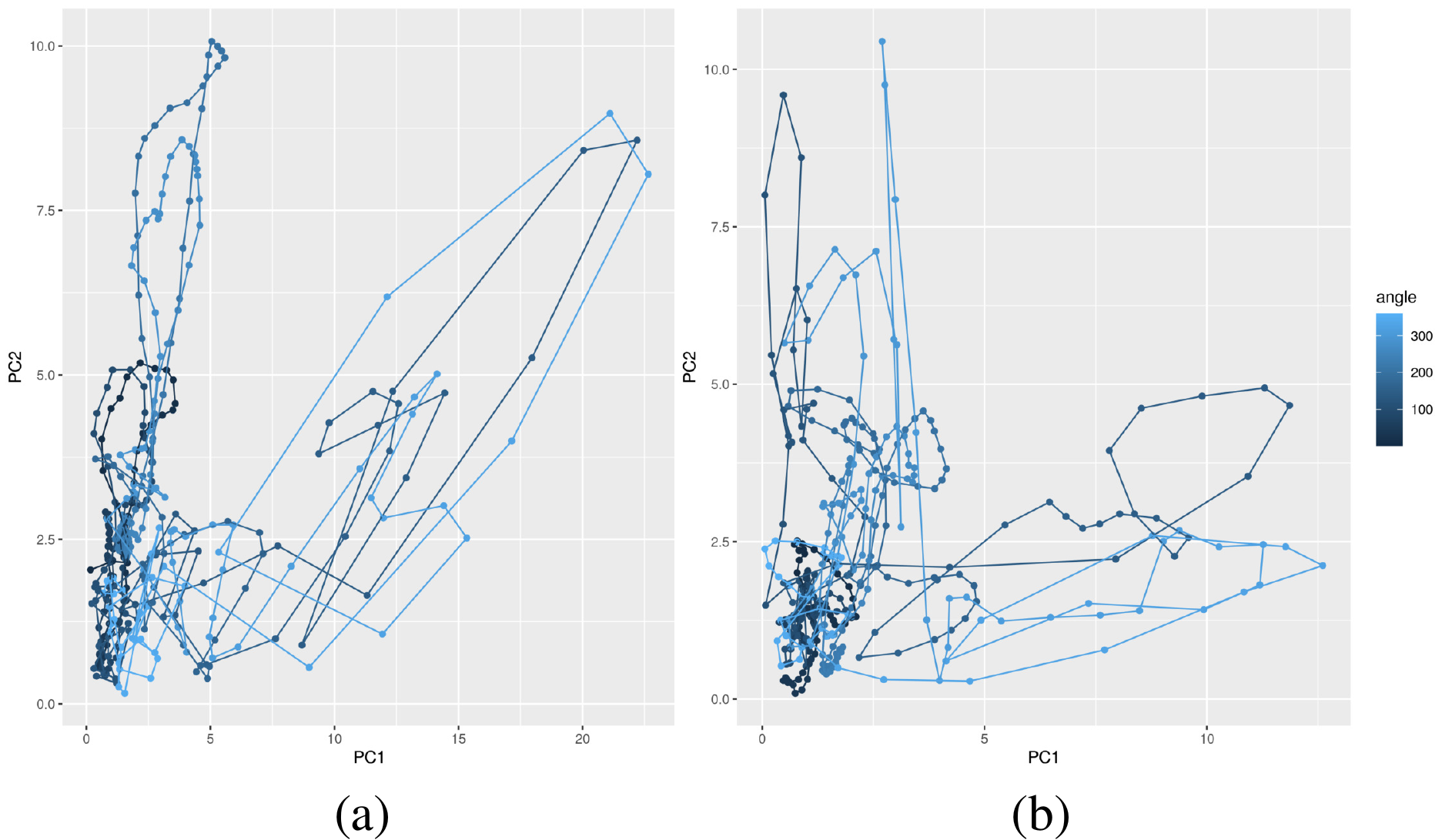}
    \caption{PCA plot of coke bottle with open lid signature space (a) and phase space (b).}
    \label{fig:coke_open_pca_data}
  \end{center}
\end{figure}

In Figure \ref{fig:coke_open_pca_data}(a), there are excursions corresponding to the reflections at $330^\circ$ and $150^\circ$.
The largest loops correspond to reflections located at $0^\circ = 360^\circ$, $100^\circ$, $220^\circ$, and $280^\circ$.

Notice that the echoes from $150^\circ$ and $330^\circ$ have only a few pulses.  This indicates that they are quite focused in angle, and correspond to the specular reflections from the lateral faces of the bottle.  The other loops correspond to reflections that are more dispersed in angle.

The phase space PCA plot in Figure \ref{fig:coke_open_pca_data}(b) shows similar features to the signature space PCA plot in Figure \ref{fig:coke_open_pca_data}(a), though the loops are somewhat more separated.

\begin{figure}
  \begin{center}
    \includegraphics[width=5in]{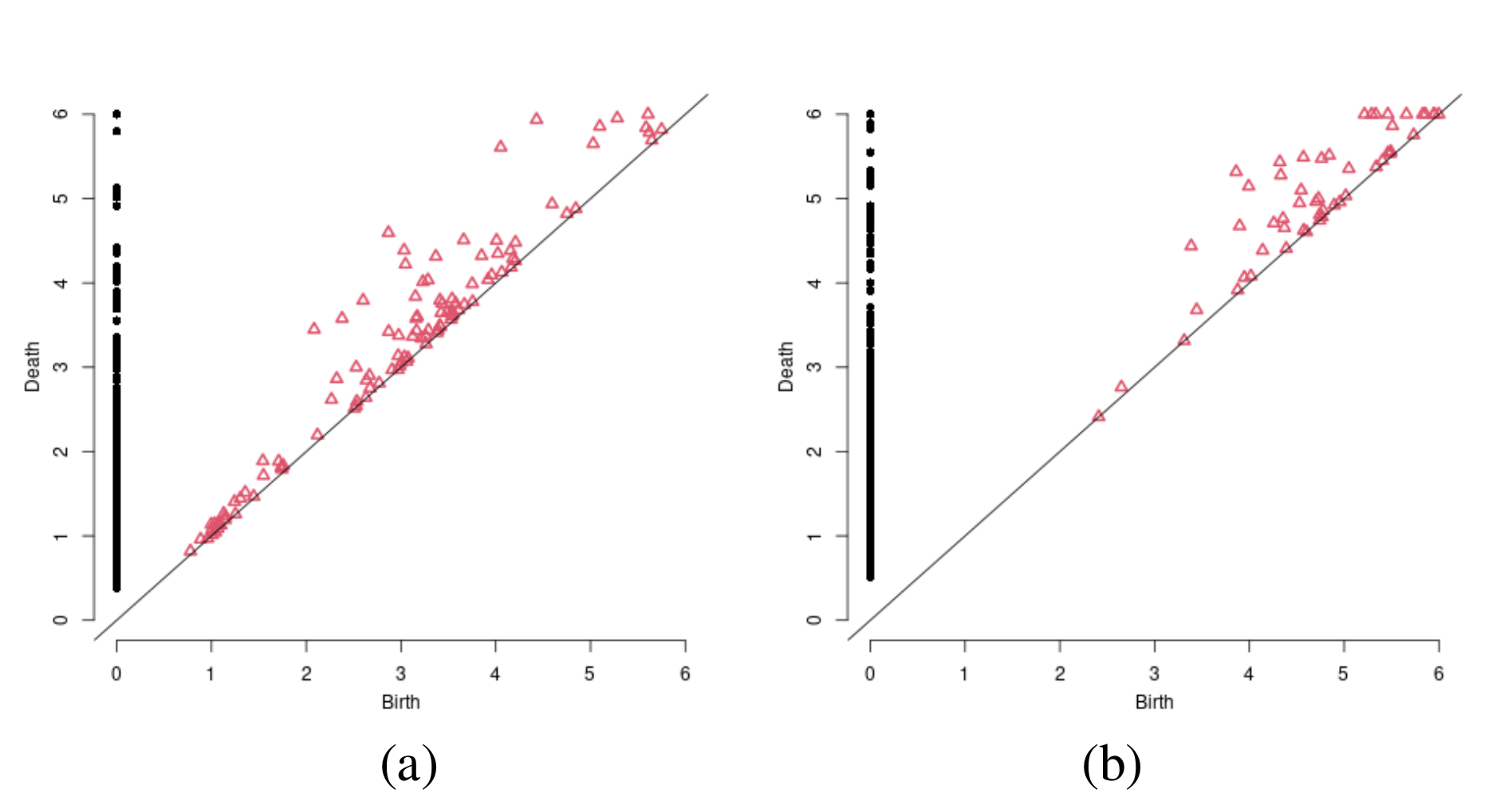}
    \caption{Persistence diagrams of coke bottle with open lid signature space (a) and phase space (b).  Black points form the diagram from $H_0$; red points form the diagram for $H_1$.}
    \label{fig:coke_open_persistence}
  \end{center}
\end{figure}

The persistence diagrams for the coke bottle with open lid are shown in Figure \ref{fig:coke_open_persistence}, but are difficult to interpret.  This is likely due to the angular dispersion, since the loops associated with the specular reflections are smaller.  As a result, they have shorter lifetimes, making them hard to distinguish from the noise.

\subsection{Coke bottle with capped lid}

\begin{figure}
  \begin{center}
    \includegraphics[width=3in]{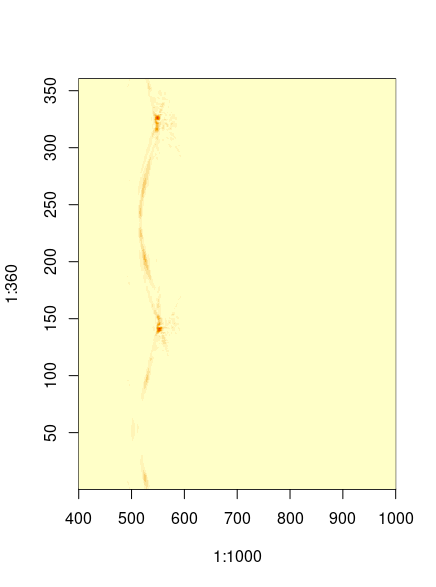}
    \caption{The AirSAS data for coke bottle with capped lid: the horizontal axis is range samples, the vertical axis is look angle in degrees.}
    \label{fig:coke_capped_raw_data}
  \end{center}
\end{figure}

In Figure \ref{fig:coke_capped_raw_data}, there are two main spectral reflections that occur at $150^\circ$ and 
$330^\circ$. These echoes result from when the sound hits the sides of the bottle and are dispersed somewhat in angle.

The strength of these reflections is comparable to that of the uncapped coke bottle in Figure \ref{fig:coke_open_raw_data}.
The reflections are still less than that of those from the cup in either configuration, which may be due to differences in the material.
Again, there are two possible diffuse reflections between $0^\circ - 150^\circ$ and $150^\circ - 330^\circ$.

\begin{figure}
  \begin{center}
    \includegraphics[width=5in]{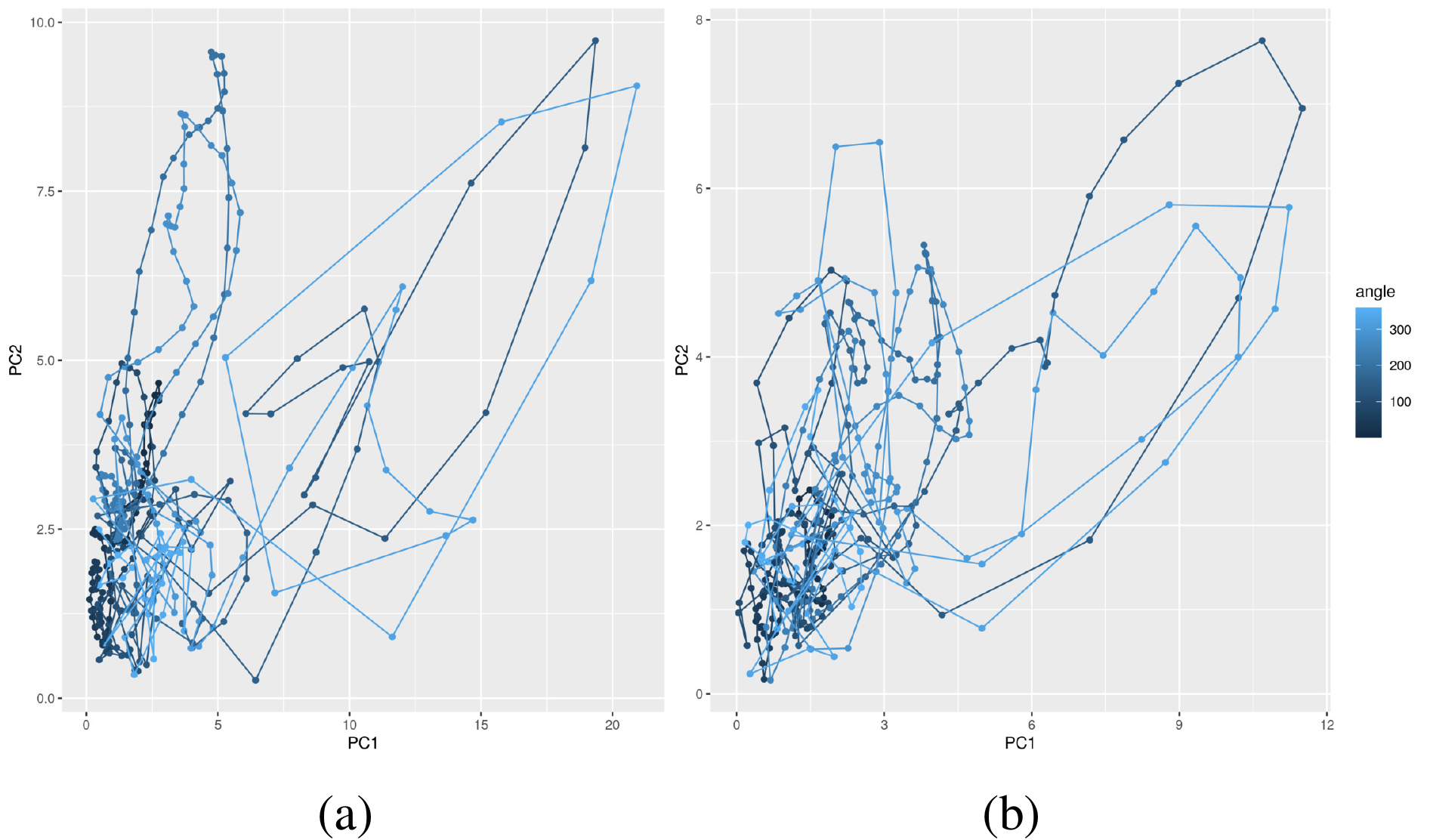}
    \caption{PCA plot of coke bottle with capped lid signature space (a) and phase space (b).}
    \label{fig:coke_capped_pca_data}
  \end{center}
\end{figure}

In Figure \ref{fig:coke_capped_pca_data}, one can observe a similar dichotomy between the sparse loops corresponding to specular reflections with small beam width at $150^\circ$ and $330^\circ$.  These echoes appear to be somewhat stronger than in the open bottle case.
The largest loops correspond to reflections at $0^\circ=360^\circ$, $100^\circ$, $220^\circ$, and $280^\circ$ approximately. 

\begin{figure}
  \begin{center}
    \includegraphics[width=5in]{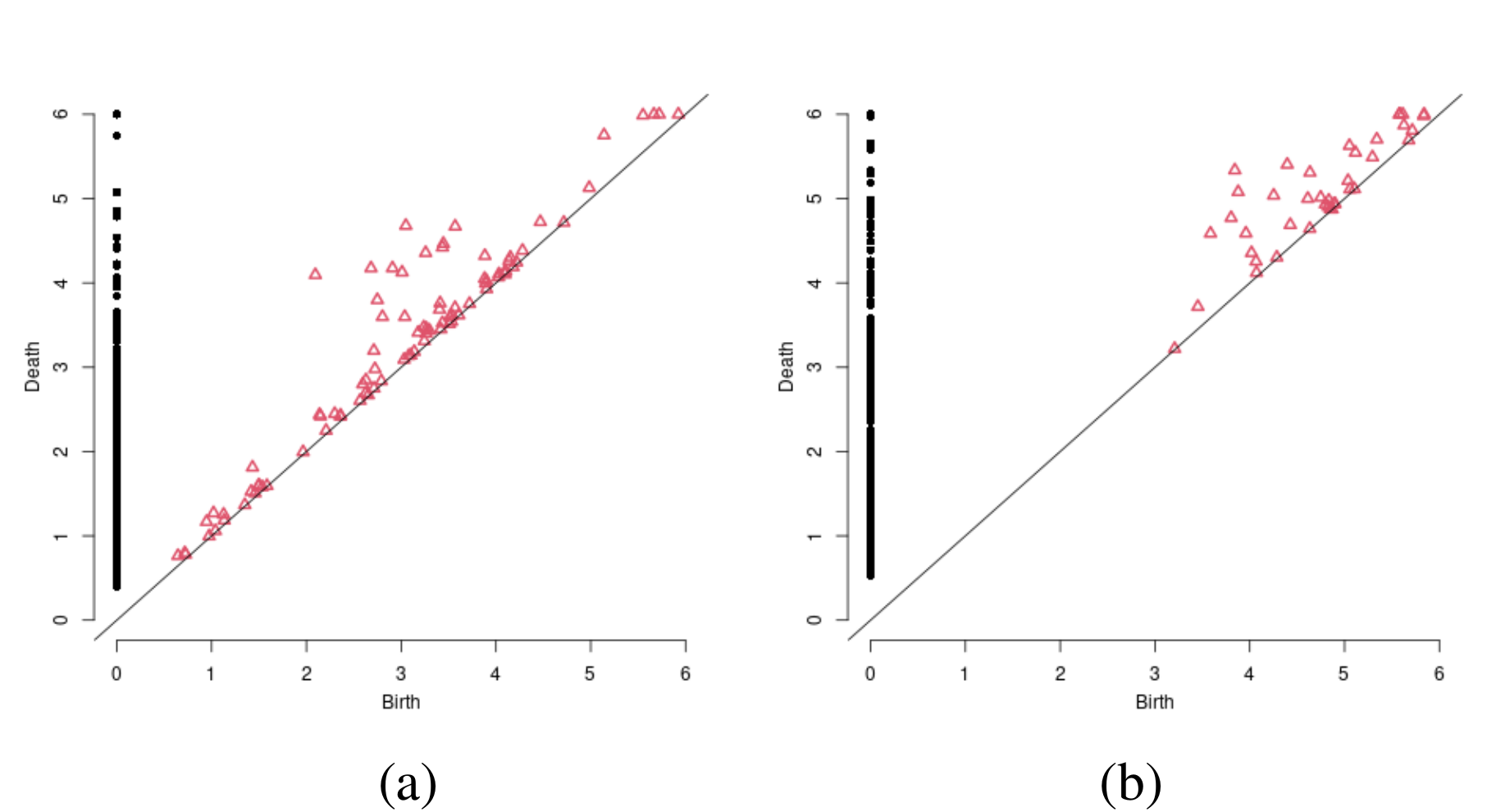}
    \caption{Persistence diagrams of coke bottle with capped lid signature space (a) and phase space (b).  Black points form the diagram from $H_0$; red points form the diagram for $H_1$.}
    \label{fig:coke_capped_persistence}
  \end{center}
\end{figure}

The persistence diagrams for the coke bottle with capped lid are shown in Figure \ref{fig:coke_capped_persistence}, but are difficult to interpret.  This is likely due to the angular dispersion.

\section{Conclusion and recommendations}

This report summarizes several qualitative features that are present in the structure of the space of sonar echoes for CSAS collections.  
There are distinct topological and geometric features in the signature and phase spaces are that are reflected in their corresponding persistence diagrams.
These features appear to correspond directly to the raw sonar data, and by extension relate to the physical acoustics of the targets.
The success of this analysis in providing valid structural information about topological features may serve as a foundation for further research.

\section*{Acknowledgments}

The authors thank Daniel Park (ARL/PSU) and his team for the use of the AirSAS data presented in this report.
 This article is based upon work supported by the Office of Naval Research (ONR) under Contract Nos. N00014-15-1-2090 and N00014-18-1-2541. Any opinions, findings and conclusions or recommendations expressed in this article are those of the author and do not necessarily reflect the views of the Office of Naval Research.

\bibliographystyle{plainnat}
\bibliography{csas_ql_bib}
\end{document}